\documentclass[11pt]{article}
\usepackage[a4paper,margin=26mm]{geometry}
\usepackage[T1]{fontenc}
\usepackage{lmodern,amsmath,amssymb,amsthm,mathtools,booktabs,microtype}
\usepackage{xcolor,array,tabularx,flafter,tikz}
\usepackage[colorlinks=true,linkcolor=blue!55!black,citecolor=blue!55!black,urlcolor=blue!55!black]{hyperref}
\usepackage{aliascnt}
\usepackage[nameinlink,noabbrev,capitalise]{cleveref}
\newtheorem{theorem}{Theorem}[section]
\newaliascnt{lemma}{theorem}
\newtheorem{lemma}[lemma]{Lemma}
\aliascntresetthe{lemma}
\newaliascnt{proposition}{theorem}

\aliascntresetthe{proposition}
\newaliascnt{corollary}{theorem}

\aliascntresetthe{corollary}
\theoremstyle{definition}
\newaliascnt{definition}{theorem}
\newtheorem{definition}[definition]{Definition}
\aliascntresetthe{definition}
\theoremstyle{remark}
\newaliascnt{remark}{theorem}

\aliascntresetthe{remark}
\DeclareMathOperator{\TV}{TV}

\newcommand{\cc}{\boxtimes}
\newcommand{\E}{\mathbb E}
\newcommand{\perfect}{\boldsymbol{\Delta}_\infty}
\newcommand{\corr}{\mathcal A}
\newcommand{\Ft}{\mathcal F_{\mathrm t}}
\newcommand{\Fnt}{\mathcal F_{\mathrm{nt}}}

\newcommand{\bms}{\mathrm{BMS}}
\newcommand{\dd}{\mathrm d}
\title{Spatially Coupled MacKay--Neal Codes\\on Symmetric Markov Noise:\\ A Reduction to BMS Channels}
\author{Kenta Kasai\\\small Institute of Science Tokyo}
\date{September 21, 2026}
\hypersetup{pdfauthor={Kenta Kasai},pdftitle={Spatially Coupled MacKay--Neal Codes on Symmetric Markov Noise:  A Reduction to BMS Channels}}
\begin{document}
\maketitle
\begin{abstract}
We prove a capacity criterion for fixed-degree spatially coupled MacKay--Neal codes on additive symmetric Markov noise. The noise bit flips with probability $p$, so the channel has one parameter and capacity $1-h_2(p)$. For every integer $\ell\geq4$ with $3/\ell<1-h_2(p)$, matched finite-window BCJR and sum-product decoding admit code sequences, possibly depending on $p$, with actual rate tending to $3/\ell$ and vanishing average information-bit error. The proof parallels the GEC-to-BEC argument: a conditional-entropy correction transfers BMS fixed-point positivity to the channel with memory. The correction also covers a capacity deficit of the effective BMS channel. A two-state posterior contraction supplies the additional derivative bound needed for threshold saturation. The result uses the cited BMS positivity theorem, with an analytic proof for $\ell\geq33$ and exact interval certificates for $4\leq\ell\leq32$.
\end{abstract}
\section{Introduction}

The generalized erasure channel (GEC) proof for spatially coupled MacKay--Neal (MN) codes reduces a channel with memory to a binary erasure channel (BEC) potential comparison~\cite{OkazakiKasai2014,FukushimaOkazakiKasai2015}. We develop its nonerasure counterpart for a single, elementary channel: additive symmetric Markov noise. The comparison channel is now binary-input memoryless symmetric (BMS). The purpose is to isolate the reduction and the extra dependence estimate needed by the detector.

\subsection{Historical background}

Finite-state noise models were developed to describe bursts of errors that memoryless models cannot represent. The Gilbert and Elliott models~\cite{Gilbert1960,Elliott1963} and the subsequent capacity and coding analysis of Gilbert--Elliott channels~\cite{MushkinBarDavid1989} established the importance of exploiting temporal dependence. The BCJR algorithm provides soft-output inference on a trellis~\cite{BCJR1974}; joint iterative detection and low-density parity-check (LDPC) decoding were later analyzed for channels with memory and erasure noise~\cite{PfisterSiegel2008}. Trellis-based rate estimation made finite-state information rates numerically accessible~\cite{ArnoldLoeliger2001,PfisterSoriagaSiegel2001}. Joint detector--decoder analysis, often called turbo equalization, established density evolution and concentration for intersymbol-interference (ISI) channels~\cite{KavcicMaMitzenmacher2003} and density evolution for Gilbert--Elliott channels~\cite{EckfordKschischangPasupathy2005}.

Spatial coupling developed from convolutional LDPC constructions~\cite{FelstromZigangirov1999}. Threshold saturation was proved for the BEC~\cite{KudekarRichardsonUrbanke2011} and then for BMS channels~\cite{KudekarRichardsonUrbanke2013}, explaining how termination improves iterative decoding. Potential methods gave a density-valued formulation of this mechanism~\cite{Kumar2014}. For channels with memory, numerical evidence of saturation~\cite{KudekarKasai2011} and studies of intersymbol-interference channels~\cite{NguyenEtAl2012} motivated the corresponding detector--decoder analysis.

MN codes originated in sparse-matrix constructions~\cite{MacKayNeal1995,MacKay1999}. Their early capacity results used optimal decoding, including for symmetric stationary ergodic noise~\cite{MacKay1999}. Spatially coupled MN constructions subsequently showed numerical thresholds near capacity with bounded degrees~\cite{KasaiSakaniwa2011}. Rigorous results followed for the BEC~\cite{OkazakiKasai2014} and GECs with memory~\cite{FukushimaOkazakiKasai2015}; the latter used the BEC potential comparison that motivates this paper. The recent fixed-degree BMS result~\cite{Kasai2026} supplies the nonerasure comparison theorem. Our contribution is to connect it to BCJR detection on a channel with memory through the entropy correction and a summable detector-interaction bound.

\subsection{Channel and result}

\begin{definition}[Symmetric Markov noise]\label{def:channel}

Let $0<p<1/2$. Independently of the entire input, let $E_1$ be fair and let $Z_2,Z_3,\ldots$ be iid Bernoulli$(p)$, independent of $E_1$. The channel is

\begin{equation}\label{eq:channel}
 E_t=E_{t-1}\oplus Z_t\quad(t\geq2),\qquad Y_t=X_t\oplus E_t,
 \qquad Q=\begin{pmatrix}1-p&p\\p&1-p\end{pmatrix}.
\end{equation}

Here $X_t,Y_t,E_t$ are the input, output, and error bits at time $t$; $Z_t$ indicates a noise-state change, $p$ is its probability, and $Q$ is the noise transition matrix. The receiver observes $Y$ and knows $p$; the noise is not revealed. There is no feedback. The stationary noise runs continuously throughout transmission.

\end{definition}

The state in the additive channel of \cref{def:channel} is the error bit itself. Write $E^n=(E_1,\ldots,E_n)$, and use the same superscript convention for other sequences. With binary entropy $h_2(p)=-p\log_2p-(1-p)\log_2(1-p)$, channel capacity $C$, and normalized block information $C_n$, we have

\begin{equation}\label{eq:capacity}
 H(E^n)=1+(n-1)h_2(p),\qquad C=1-h_2(p),
 \qquad C_n:=1-\frac{H(E^n)}n=\left(1-\frac1n\right)C.
\end{equation}

Indeed, $I(X^n;Y^n)\leq n-H(E^n)$ with equality for iid fair inputs. For iid fair inputs the normalized information density is $1+n^{-1}\log_2 P(E^n)$, which converges almost surely to $C$. The direct coding theorem~\cite[Theorem~2]{VerduHan1994} and the preceding mutual-information upper bound therefore give the operational capacity, also covered by the stationary-ergodic additive-noise result~\cite[Corollary~2]{Alajaji1995}. Thus capacity equals the symmetric information rate. Each isolated output bit conveys no information about its input, since $E_t$ is fair. Information comes from temporal dependence: $Y_t\oplus Y_{t-1}=X_t\oplus X_{t-1}\oplus Z_t$. This makes the role of the detector explicit.

In $(\ell,r,g)$-MN notation, punctured variables have degree $\ell$, transmitted variables have degree $g$, and each check has $r$ punctured and $g$ transmitted sockets. We use $r=g=3$; \cref{sec:codes} specifies the ensemble.

\begin{theorem}[Finite-code capacity criterion]\label{thm:main}

Fix the symmetric Markov-noise channel of \cref{def:channel} and an integer $\ell\geq4$ whose design rate $R=3/\ell$ is below the channel capacity $C$ in \cref{eq:capacity}. The terminated, uniformly smoothed $(\ell,3,3)$-MN ensembles of \cref{sec:codes} admit deterministic realizations whose actual transmitted rate tends to $R$ and whose average information-bit error tends to zero under matched finite-window BCJR detection and sum-product decoding (SPA). The average includes uniform messages, channel noise, and fair random tie-breaking. The construction may depend on $p$. The degrees remain fixed; width, chain length, iteration count, window radius, and section size are selected in that order.

\end{theorem}

The proof uses the BMS results summarized in \cref{sec:prelim}, from the preprint~\cite{Kasai2026}. Its proof is analytic for $\ell\geq33$ and uses exact interval certificates for $4\leq\ell\leq32$. This dependency remains; the present paper proves the transfer to Markov noise. The result concerns the discrete rates $3/\ell$ below capacity, not arbitrary rates at one fixed degree.

\subsection{Proof outline and relation to the GEC argument}

The capacity criterion in \cref{thm:main} follows in three steps. First, separate the perfect-decoding point, evaluate the trivial fixed-point set $\Ft$, and compare the nontrivial fixed-point set $\Fnt$ with BMS fixed points; these sets are defined in \cref{def:fixedpoints}. As in the GEC proof, the trivial point has potential equal to information rate minus code rate; here its channel-parametrized branch collapses to one point in density space. The GEC area correction becomes a conditional-entropy correction $\corr$, giving the potential decomposition $U=U_\bms+\corr$ in \cref{eq:reduction}. Second, prove that $\corr$ compensates even when the effective BMS channel has capacity below $R$. The correction bound of \cref{lem:correction} and the BMS comparison of \cref{lem:bms} yield the uniform positive gap of \cref{thm:gap}. Third, use the posterior contraction of \cref{lem:contraction}, with factor $\rho=1-2p$, to bound detector interactions. The mixed-derivative bound in \cref{eq:hessian} then supplies the additional term in the threshold-saturation proof of \cref{thm:saturation}. The first two steps form \cref{sec:reduction}; the third uses \cref{sec:memory,sec:coupling}. Finite-window density evolution and the rate argument in \cref{sec:codes} then yield finite codes.

Thus the analogy is a reduction of potential positivity, not an identification of the full memory-channel recursion with one fixed memoryless channel. The GEC proof uses BEC positivity; this proof uses BMS positivity and one additional contraction estimate.

\Cref{sec:extensions} separates the parts of this argument that extend beyond the present channel from the detector properties that require further proof.

\section{Preliminaries}\label{sec:prelim}

A BMS law $d$ is represented canonically as a mixture of binary symmetric channels (BSCs) with the crossover probability $\varepsilon\in[0,1/2]$ revealed. Its reliability is $V=1-2\varepsilon$, its conditional entropy is $H(d)=\E h_2(\varepsilon)$, its capacity is $C(d)=1-H(d)$, and its Bhattacharyya parameter is $Z(d)=\E\sqrt{1-V^2}$. Thus $C(d)$ concerns a BMS law, whereas $C$ in \cref{eq:capacity} is the physical Markov-channel capacity. The laws $\Delta_0$ and $\Delta_\infty$ denote absent and perfect information. Variable convolution $a\star b$ combines independent observations of one bit; check convolution $a\cc b$ observes a parity. Powers $d^{\star k}$ and $d^{\cc k}$ denote $k$ repeated variable and check convolutions, respectively. We use the standard BMS calculus~\cite[Sec.~II]{Kumar2014}:

\begin{equation}\label{eq:duality}
 H(a\star b)+H(a\cc b)=H(a)+H(b),\qquad
 d\star\Delta_0=d,\quad d\cc\Delta_\infty=d.
\end{equation}

Also $d\star\Delta_\infty=\Delta_\infty$ and $d\cc\Delta_0=\Delta_0$. Entropy extends linearly to signed measures. We write $a\preceq b$ when $a$ is more informative than $b$, and $\TV(a,b)$ for total variation of their canonical reliability laws. Mixture labels are always revealed. Weak convergence refers to laws on the compact interval $[0,1]$; the bounded-Lipschitz distance $\operatorname{dist}_{\mathrm{BL}}$, with test-function norm and Lipschitz constant at most one, metrizes this topology.

\paragraph{BMS results used.}
For density evolution (DE) of the same balanced, uniformly smoothed $(\ell,3,3)$ ensemble, $\ell\geq4$, put $R=3/\ell$ and $\perfect=(\Delta_\infty,\Delta_\infty)$. The cited BMS paper proves: for every BMS law $c$, each fixed point other than $\perfect$ has $\max\{Z(a),Z(b)\}\geq1/25$; if $C(c)>R$, its potential is positive. Uniformly over $C(c)\geq R+\delta$, $0<\delta\leq1-R$, the nonperfect fixed-point gap is positive~\cite[Lemma~6.1, Theorem~6.3, Corollary~6.4]{Kasai2026}. For a fixed $c$, an ordered update decreases potential~\cite[Lemma~7.1]{Kasai2026}. With both variable types shortened, capacity $C(c)>R$ implies convergence of both variable-type DE errors to zero on every finite chain for sufficiently large width~\cite[Corollary~7.4]{Kasai2026}. These are the imported results; the detector estimates and their finite-code realization are proved below.

\begin{definition}[Helpers and extrinsic laws]\label{def:helpers}

For analysis, use iid fair inputs $X^n$. Conditional on $X^n$, generate independent BMS observations $\Lambda_i$ of laws $d_i$, also independent of the physical output. Write $\boldsymbol d=(d_1,\ldots,d_n)$ and define

\begin{equation}\label{eq:primitive}
 F_n(\boldsymbol d)=H(X^n\mid Y^n,\Lambda^n),\qquad
 h_n(d)=\frac1nF_n(d,\ldots,d).
\end{equation}

Here $F_n$ is the remaining block uncertainty, and $h_n$ its per-bit value for common helper law $d$. The extrinsic experiment at $i$ observes $(Y^n,\Lambda_{-i})$, where $\Lambda_{-i}$ contains all helpers except $i$; $\Lambda_{<i},\Lambda_{>i}$ contain those before and after $i$. Let $e_n(d)$ be its coordinate average with the coordinate label revealed. Let $e_W^{\mathrm{win}}(d)$ be the analogous stationary detector at the center of a radius-$W$ window.

\end{definition}

The helpers in \cref{def:helpers} represent incoming SPA messages on a computation tree. Omitting the root helper prevents the detector from directly reusing that incoming observation. For the length-$n$ channel law $W_n(y\mid x)=\Pr(Y^n=y\mid X^n=x)$, the identity $W_n(y\oplus v\mid x\oplus v)=W_n(y\mid x)$ makes every extrinsic experiment BMS and permits the all-zero-codeword analysis. No transmitted dither is needed.

The matched BCJR detector~\cite{BCJR1974} has two states. Given a prior $\nu_t$ for $X_t$, use the likelihood $D_t=\operatorname{diag}\{\nu_t(y_t),\nu_t(y_t\oplus1)\}$. The forward and backward recursions multiply by $D_tQ$ and $QD_t$, respectively, with stationary left boundary and an all-one right vector. At the root omit $D_t$. Neither endpoint noise bit is revealed.

We use the limits

\begin{equation}\label{eq:limits}
 h(d)=\lim_{n\to\infty}h_n(d),\qquad e(d)=\lim_{W\to\infty}e_W^{\mathrm{win}}(d)
                  =\lim_{n\to\infty}e_n(d).
\end{equation}

The limiting conditional entropy $h(d)$ and extrinsic law $e(d)$ in \cref{eq:limits} exist uniformly and are continuous by the detector regularity result, \cref{lem:regularity}. In particular, $h(\Delta_0)=1-C$, $h(\Delta_\infty)=0$, and $e(\Delta_0)=\Delta_0$. The last equality follows because complementing all input and noise bits preserves every observation law when no helpers are present.

\section{The BMS reduction}\label{sec:reduction}

\subsection{The counterpart of the GEC area correction}

\begin{lemma}[Entropy correction]\label{lem:correction}

For the independent helpers of \cref{def:helpers} with common BMS law $d$, let $c=e(d)$ be the limiting extrinsic law from \cref{eq:limits}. The entropy correction satisfies

\begin{equation}\label{eq:correction}
 \corr(d):=H(c\star d)-h(d)\geq\max\{0,C-C(c)\}.
\end{equation}

\end{lemma}\begin{proof}

First work in a finite block and put $c_n=e_n(d)$. Conditional independence of the local helper and extrinsic observation gives

\begin{equation}\label{eq:totalcorrelation}
 \corr_n(d):=H(c_n\star d)-h_n(d)
 =\frac1n\left\{\sum_i H(X_i\mid Y^n,\Lambda^n)-H(X^n\mid Y^n,\Lambda^n)\right\}\geq0.
\end{equation}

The correction in \cref{eq:totalcorrelation} is posterior total correlation per bit~\cite{Watanabe1960}. A stronger bound follows from the helper chain rule:

\begin{align}\label{eq:helpergain}
 I(X^n;\Lambda^n\mid Y^n)
 &=\sum_i I(X_i;\Lambda_i\mid Y^n,\Lambda_{<i})\nonumber\\
 &\geq\sum_i I(X_i;\Lambda_i\mid Y^n,\Lambda_{-i}).
\end{align}

The difference of the summands is $I(\Lambda_i;\Lambda_{>i}\mid Y^n,\Lambda_{<i})\geq0$, since a local helper depends only on its input bit. Dividing by $n$ gives $h_n(\Delta_0)-h_n(d)\geq H(c_n)-H(c_n\star d)$, or $\corr_n(d)\geq C_n-C(c_n)$. Use the entropy and detector limits in \cref{eq:limits} and the block-information limit $C_n\to C$ in \cref{eq:capacity} to obtain the capacity-deficit bound in \cref{eq:correction}.

\end{proof}

To see the GEC correspondence, vary one helper through a revealed mixture $d_i(t)=(1-t)u+t v$. Let $e_i$ denote the extrinsic BMS law at coordinate $i$ with all other helper laws held fixed. Write $Dh_n(d)[v-u]$ for the derivative at $d$ in the signed direction $v-u$, with the same perturbation applied at every helper position. The chain rule, conditioning first on $X_i$, gives

\begin{equation}\label{eq:variation}
 \frac{\dd F_n}{\dd t}=H(e_i\star(v-u)),\qquad
 Dh_n(d)[v-u]=H(e_n(d)\star(v-u)).
\end{equation}

With erasure helpers and $\phi_n(x)=H(e_n(\mathrm{BEC}(x)))$, integration from perfect helpers yields

\begin{equation}\label{eq:gec}
 h_n(\mathrm{BEC}(x))=\int_0^x\phi_n(u)\,\dd u,\qquad
 \corr_n(\mathrm{BEC}(x))=x\phi_n(x)-\int_0^x\phi_n(u)\,\dd u.
\end{equation}

For a GEC, $\phi_n$ is the extrinsic erasure probability, recovering the area correction~\cite[Lemma~3]{FukushimaOkazakiKasai2015}. The entropy-correction bound in \cref{eq:correction} is its information-theoretic replacement for nonerasure messages.

\subsection{Fixed points and potential}

Throughout, $r=g=3$, $\ell\geq4$, and $R=r/\ell$. Let $x=(a,b)$ be the variable-to-check BMS laws for punctured and transmitted variables, respectively. The check outputs are $q_1,q_2$, and $v$ is the parity-observation law of a whole check. The pair $y=(y_1,y_2)$ denotes incoming check-message laws for the two variable types. Retaining $r,g$ only to identify socket types, define the check update $\mathsf g$, variable-and-detector update $\mathsf f$, and full update $\mathsf T$ by

\begin{align}\label{eq:maps}
 q_1&=a^{\cc(r-1)}\cc b^{\cc g},&q_2&=a^{\cc r}\cc b^{\cc(g-1)},&
 v&=a\cc q_1=b\cc q_2,\nonumber\\
 \mathsf g(a,b)&=(q_1,q_2),&\mathsf f(y)&=(y_1^{\star(\ell-1)},e(y_2^{\star g})\star y_2^{\star(g-1)}),&
 \mathsf T&=\mathsf f\circ\mathsf g.
\end{align}
\begin{align}\label{eq:potential}
 \mathcal Q(a,b)&=-(r+g-1)H(v)+rH(q_1)+gH(q_2),\nonumber\\
 \mathcal P(y)&=R H(y_1^{\star\ell})+h(y_2^{\star g}),\qquad
 U=\mathcal Q-\mathcal P\circ\mathsf g.
\end{align}

In \cref{eq:potential}, $\mathcal Q$ collects the check contributions, $\mathcal P$ the variable and detector contributions, and $U$ is the uncoupled potential. The ordinary BMS potential~\cite[Sec.~4]{Kasai2026} replaces $h(d)$ by $H(c\star d)$. Consequently,

\begin{equation}\label{eq:reduction}
 U(a,b)=U_\bms(a,b;c)+\corr(d),\qquad d=q_2^{\star g},\quad c=e(d).
\end{equation}

\begin{definition}[Trivial and nontrivial fixed points]\label{def:fixedpoints}

Separate the perfect-decoding point $\perfect=(\Delta_\infty,\Delta_\infty)$ and define

\begin{equation}\label{eq:fixedpoints}
 \mathcal F=\{x\ne\perfect:\mathsf T(x)=x\},\qquad
 \Ft=\{(\Delta_0,\Delta_0)\},\qquad \Fnt=\mathcal F\setminus\Ft.
\end{equation}

Thus $\mathcal F$ contains only nonperfect fixed points. Its trivial part is $\Ft$ and its nontrivial part is $\Fnt$, following the classification in the GEC proof~\cite[Sec.~III-A]{FukushimaOkazakiKasai2015}.

\end{definition}

The perfect point and the point in $\Ft$ are fixed. If $a=\Delta_0$, then $q_1=q_2=\Delta_0$ and $b'=e(\Delta_0)=\Delta_0$, so there are no other fixed points with $a=\Delta_0$. Direct substitution gives

\begin{equation}\label{eq:trivialvalues}
 U(\perfect)=0,\qquad U(\Delta_0,\Delta_0)=1-R-(1-C)=C-R.
\end{equation}

In the GEC proof, with detector erasure transfer $\phi$ and channel parameter $\epsilon$, the trivial point is $(x_1,x_2)=(1,\phi(1;\epsilon))$~\cite[Sec.~III-A, Lemma~1]{FukushimaOkazakiKasai2015}. For each fixed channel it is one point; varying $\epsilon$ gives a one-parameter branch on $x_1=1$. For the BEC this is $(1,\epsilon)$. Its density-valued counterpart is $(\Delta_0,e_p(\Delta_0))$, where the subscript makes the channel parameter explicit. Here $e_p(\Delta_0)=\Delta_0$ for every $p$: without helpers, simultaneous complementation of the input and noise leaves the observation and their joint law unchanged. Every individual input bit therefore remains fair given the output. The branch collapses to one point in density space, although $(\Delta_0,\Delta_0,p)$ still varies with $p$.

Its role in the proof is unchanged: the GEC trivial-point potential is $I(\epsilon)-R$, where $I$ is the symmetric information rate; the trivial-point evaluation in \cref{eq:trivialvalues} gives $C(p)-R$ here. After separating this value from $U(\perfect)=0$, only $\Fnt$ needs the BMS comparison. At the uninformative point the frozen BMS channel has capacity zero, so $U_\bms=-R$ while $\corr=C$. This example within the model shows why nonnegativity of $\corr$ alone would not justify applying a BMS theorem that requires capacity above $R$.

\subsection{Transferring the positive gap}

\begin{lemma}[BMS potential comparison]\label{lem:bms}

Let $\mathsf T_c$ be the BMS update obtained from the maps in \cref{eq:maps} by replacing the detector output $e(y_2^{\star g})$ by a fixed channel law $c$. For a comparison capacity $R<\tau<1$, the BMS MN theorem gives the uniform potential gap

\begin{equation}\label{eq:gamma}
 \Gamma_\ell(\tau):=\inf\{U_\bms(x;c):\mathsf T_c(x)=x,\ x\ne\perfect,\ C(c)\geq\tau\}>0.
\end{equation}

For any BMS fixed point $x$ with $C(c)\leq\tau$, one also has $U_\bms(x;c)\geq C(c)-\tau$.

\end{lemma}\begin{proof}

The first assertion is exactly \cite[Corollary~6.4]{Kasai2026}, with $\delta=\tau-R$.

For the second assertion, improve $c$ by the revealed mixture $c_\tau=(1-\lambda)c+\lambda\Delta_\infty$, where $\lambda=(\tau-C(c))/(1-C(c))$. Then $\mathsf T_{c_\tau}(x)\preceq x$. Ordered descent~\cite[Lemma~7.1]{Kasai2026} and weak continuity give a limiting fixed point of nonnegative potential, hence $U_\bms(x;c_\tau)\geq0$: the limit has potential zero if it is $\perfect$, and positive potential otherwise since $C(c_\tau)=\tau>R$. At the original point $x$, keep the combined helper law $d=q_2^{\star g}$ from the potential decomposition \cref{eq:reduction} fixed. Only the channel contribution to the potential changes, and

\begin{equation}\label{eq:upgrade}
 U_\bms(x;c_\tau)-U_\bms(x;c)=\lambda H(c\star d)
 \leq\lambda H(c)=\tau-C(c).
\end{equation}

This proves the assertion.

\end{proof}\begin{theorem}[Uniform positive potential gap]\label{thm:gap}

If $C>R$, every fixed point other than $\perfect$ has potential at least

\begin{equation}\label{eq:gap}
 \Delta:=\min\{C-\tau,\Gamma_\ell(\tau)\}>0,\qquad \tau=(C+R)/2.
\end{equation}

Equivalently, $\min\{C-R,\inf_{x\in\Fnt}U(x)\}\geq\Delta$, with $\inf\varnothing=+\infty$.

\end{theorem}\begin{proof}

The trivial fixed-point set $\Ft$ of \cref{def:fixedpoints} has potential $C-R$ by the explicit values in \cref{eq:trivialvalues}. For a nontrivial fixed point $x\in\Fnt$, freeze its extrinsic law $c=e(q_2^{\star g})$. The update equations in \cref{eq:maps} then give $\mathsf T_c(x)=x$: this particular point is also a BMS fixed point. The comparison law $c$ may differ between points. If its capacity $C(c)\geq\tau$, the potential decomposition in \cref{eq:reduction} and the BMS gap in \cref{eq:gamma} give $U(x)\geq\Gamma_\ell(\tau)$. Otherwise, the low-capacity BMS bound of \cref{lem:bms} may be negative, but the entropy correction of \cref{lem:correction} compensates for it:

\begin{equation}\label{eq:compensate}
 U(x)\geq[C(c)-\tau]+[C-C(c)]=C-\tau.
\end{equation}

The comparison capacity $\tau=(C+R)/2$ in \cref{eq:gap} lies strictly between the design rate and physical capacity. Thus the compensated bound $C-\tau$ in \cref{eq:compensate} is positive, and taking the smaller of the two bounds gives the same gap $\Delta$ at every nonperfect fixed point.

\end{proof}

\section{The one estimate needed for memory}\label{sec:memory}

The positive-gap theorem, \cref{thm:gap}, supplies the potential bound needed for saturation. To use the spatial-shift argument, we must also control how the detector changes when its helpers change. For the symmetric Markov channel of \cref{def:channel}, define a one-step contraction factor $\rho$ and a bound $J_p$ on the total influence of all other positions:

\begin{equation}\label{eq:rho}
 \rho=1-2p\in(0,1),\qquad J_p=2\sum_{k\geq1}\rho^k=\frac{2\rho}{1-\rho}.
\end{equation}
\begin{lemma}[Posterior contraction]\label{lem:contraction}

Use the contraction factor $\rho$ of \cref{eq:rho}. The Dobrushin coefficient is the largest total variation distance between two conditional output laws obtained by changing the input bit. After conditioning on the output and arbitrary local helpers, the channel from one bit to a bit $k$ positions away has Dobrushin coefficient at most $\rho^k$, on its conditional support. The same holds in the reverse direction and with perfect helpers.

\end{lemma}\begin{proof}

Given $Y$, write $X_t=Y_t\oplus E_t$. The posterior is a binary chain with edge factors given by the transition matrix $Q$ in \cref{eq:channel} and nonnegative local likelihoods. Absorb all likelihoods to the right of a transition into a vector proportional to $(1,s)$, $0\leq s\leq\infty$. The difference of the two transition rows, measured by their probabilities of state one, is

\begin{equation}\label{eq:twostate}
 \frac{(1-p)s}{p+(1-p)s}-\frac{ps}{1-p+ps}
 =\frac{\rho s}{[p+(1-p)s][1-p+ps]}\leq\rho.
\end{equation}

For $s>0$, the denominator divided by $s$ is $p(1-p)(s+s^{-1})+p^2+(1-p)^2\geq1$; the endpoints follow by continuity. The Dobrushin coefficient of a composition is bounded by the product of the one-step coefficients. Since $Q$ is symmetric, the same calculation applies backward. A perfect local helper only removes a state from the support. Conditioning on a source value cancels its own likelihood, and the endpoint transformations between $E$ and $X$ preserve total variation.

\end{proof}\begin{lemma}[Detector limits and interaction bound]\label{lem:regularity}

The entropy and detector limits in \cref{eq:limits} are uniform and weakly continuous, and the extrinsic map $e$ preserves the degradation order defined in \cref{sec:prelim}. More precisely,

\begin{equation}\label{eq:uniform}
 0\leq h_n(d)-h(d)\leq\frac{C}{n},\qquad
 \sup_d\operatorname{dist}_{\mathrm{BL}}(e_W^{\mathrm{win}}(d),e(d))\leq4\rho^W.
\end{equation}

For BMS laws $u,v,a,b$ with $a\preceq b$, the mixed derivative of the normalized entropy $h_n$ from \cref{eq:primitive} satisfies the following bound, with the summed-interaction constant $J_p$ from \cref{eq:rho}:

\begin{equation}\label{eq:hessian}
 |D^2h_n(d)[u-v,a-b]|\leq J_p\TV(u,v)[H(b)-H(a)].
\end{equation}

\end{lemma}\begin{proof}

Write $O_L=(Y_L,\Lambda_L)$ and $O_R=(Y_R,\Lambda_R)$ for adjacent blocks of lengths $m,n$. The entropy chain rule gives
\begin{equation}\label{eq:cut}
 \begin{aligned}
 F_m+F_n-F_{m+n}
 &=I((X_L,O_L);(X_R,O_R))-I(O_L;O_R)\\
 &=I(E_L;E_R)-I(O_L;O_R)=C-I(O_L;O_R)\in[0,C].
 \end{aligned}
\end{equation}
The second equality uses $E=X\oplus Y$, iid inputs, and local independent helpers. Markovity gives $I(E_L;E_R)=I(E_m;E_{m+1})=C$; data processing gives $I(O_L;O_R)\leq C$. Subadditivity yields $h=\inf_n F_n/n$, while $F_{kn}\geq kF_n-(k-1)C$ yields $h\geq h_n-C/n$. This proves the entropy approximation in \cref{eq:uniform}; the values with absent or perfect helpers follow from the block-information formula in \cref{eq:capacity}.

On a common bilateral experiment, local and full root posteriors differ only through the two boundary conditions. Changing one boundary changes the root probability by at most $\rho^W$, by the posterior contraction in \cref{lem:contraction}; changing both costs at most $2\rho^W$. The canonical reliability $|2\Pr(X_0=0\mid\cdot)-1|$ changes by at most twice this amount. Bounded martingale convergence defines the full posterior and proves the window-detector approximation in \cref{eq:uniform}. A root at distance $W$ from a finite block's ends obeys the same bound; exceptional roots have fraction at most $2W/n$. Hence

\begin{equation}\label{eq:blocklimit}
 \sup_d\operatorname{dist}_{\mathrm{BL}}(e_n(d),e(d))
 \leq8\rho^W+4W/n.
\end{equation}

For a fixed window, condition on the helper crossover probabilities in $[0,1/2]$. Every pattern of physical outputs and helper signs has positive probability, even with perfect helpers: some input string satisfies the helper signs, and every noise string has positive probability. The normalization is continuous and bounded away from zero on this compact parameter set. Thus posterior reliabilities and their bounded continuous test integrals depend continuously on the helper laws. The same holds for entropy, since $h_2$ is continuous including its endpoints. Uniform limits preserve continuity. Garbling the helpers proves degradation preservation.

For the derivative bound, retain different helpers at coordinates $i\ne j$. Let $A_i$ be the helper at $i$ and let $V_0=(Y^n,\Lambda_{-\{i,j\}})$ collect all observations except the two helpers. Couple the more informative helper $B_j^+$, of law $d_j^+$, to its degraded version $B_j^-$, of law $d_j^-$, as $X_j\to B_j^+\to B_j^-$. Expanding conditional entropy shows that its mixed difference is the change, under the helper at $i$, of

\begin{equation}\label{eq:interaction}
 J(A_i)=I(A_i;B_j^+\mid V_0,B_j^-),\qquad
 0\leq J(A_i)\leq\rho^{|i-j|}[H(d_j^-)-H(d_j^+)].
\end{equation}

To prove the bound, condition on $(V_0,B_j^-)$ and use the chain $B_j^+\to X_j\to X_i\to A_i$. A binary-input channel whose rows have TV distance at most $\eta$ contracts mutual information by at most $\eta$: split the rows into an input-independent common part of mass $1-\eta$ and a residual part, and reveal the branch. Apply the posterior contraction of \cref{lem:contraction} with $\eta=\rho^{|i-j|}$. After averaging over the conditioning observations, conditional independence of $B_j^+$ and $V_0$ given $(X_j,B_j^-)$ gives $I(X_j;B_j^+\mid V_0,B_j^-)=I(X_j;B_j^+\mid B_j^-)-I(B_j^+;V_0\mid B_j^-)$. Dropping the nonnegative last term leaves $H(d_j^-)-H(d_j^+)$, the information gained by improving the helper at $j$. This proves the two-position interaction bound in \cref{eq:interaction}.

A canonical signed difference $u-v$ is $\TV(u,v)$ times a difference of probability measures. Affinity in revealed helper mixtures and the two-position interaction bound in \cref{eq:interaction} bound the mixed difference by $\rho^{|i-j|}\TV(u,v)[H(b)-H(a)]$. Diagonal second derivatives vanish. Sum over $i\ne j$, divide by $n$, and use $\sum_{i\ne j}\rho^{|i-j|}\leq J_p$ for each $j$. This proves the mixed-derivative bound in \cref{eq:hessian}: one direction is measured in total variation, and the ordered direction in entropy gain.

\end{proof}

\section{Threshold saturation}\label{sec:coupling}

We combine the positive potential gap of \cref{thm:gap} with the detector mixed-derivative bound in \cref{eq:hessian}. Let $L$ be the number of active variable sections and $w$ the coupling width. Using the local update maps $\mathsf g,\mathsf f$ of \cref{eq:maps}, write the smoothed recursion in check-input coordinates:

\begin{equation}\label{eq:coupled}
 y_i=\frac1w\sum_{k=0}^{w-1}\mathsf g(x_{i+k}),\qquad
 x'_j=\frac1w\sum_{k=0}^{w-1}o_{j-k},\qquad
 o_i=\begin{cases}\mathsf f(y_i),&1\leq i\leq L,\\\perfect,&\text{otherwise}.\end{cases}
\end{equation}

Here $x_j$ is the pair of laws entering checks in section $j$, $y_i$ averages check-output laws reaching variable section $i$, and $o_i$ is its outgoing pair before smoothing. The perfect pair $\perfect$ from \cref{def:fixedpoints} implements the shortened boundary.

\begin{theorem}[Threshold saturation]\label{thm:saturation}

Suppose the channel capacity satisfies $C>R$. Use the positive gap $\Delta$ from \cref{eq:gap} and summed-interaction bound $J_p$ from \cref{eq:rho}. For every finite chain length $L$, the coupled density evolution in \cref{eq:coupled} converges to perfect information from uninformative initialization whenever

\begin{equation}\label{eq:width}
 w>\frac{K_p}{2\Delta},\qquad
 K_p=(r+g-1)^2[r(\ell-1)+g(g-1)+g^2J_p]
     =25[3(\ell-1)+6+9J_p].
\end{equation}

\end{theorem}\begin{proof}

We adapt the BMS MN shift proof~\cite[Lemmas~7.1--7.2 and Theorem~7.3]{Kasai2026}, recording the detector term and the limiting step. Replace the limiting entropy $h$ and extrinsic map $e$ from \cref{eq:limits} by the finite-block quantities $h_n,e_n$ of \cref{def:helpers}. This defines a finite-block potential $U_n$ and update $\mathsf T_n$ from \cref{eq:potential,eq:maps}. These finite-block functionals are polynomials along labelled mixture paths. The entropy variation formula \eqref{eq:variation} and entropy duality \eqref{eq:duality} give

\begin{equation}\label{eq:residualscalar}
 DU_n(x)[v]=\sum_{\alpha=1}^2\omega_\alpha
 H\bigl((x_\alpha-\mathsf T_{n,\alpha}(x))\star D\mathsf g_\alpha(x)[v]\bigr),
 \qquad(\omega_1,\omega_2)=(r,g).
\end{equation}

In the residual identity \cref{eq:residualscalar}, $x=(a,b)$ is a message-law pair, $v$ a signed perturbation, and $\alpha$ the variable type. Integrate along a fixed mixture path, then use the uniform limits of \cref{lem:regularity} and continuity to replace $U_n,\mathsf T_n$ by $U,\mathsf T$. If $x\preceq\mathsf T(x)$, the path $z_s=(1-s)x+s\mathsf T(x)$ satisfies $z_s\preceq\mathsf T(x)\preceq\mathsf T(z_s)$. The residual and derivative leaves (the differentiated convolution factors) therefore have opposite order orientations. The signed-entropy inequality~\cite[Prop.~8(iii)]{Kumar2014} gives $U(\mathsf T(x))\leq U(x)$; the reverse ordering gives the same conclusion. Ordered BMS sequences converge weakly and retain their order at the limit~\cite[Prop.~11(v)--(vi)]{Kumar2014}. Update continuity makes the limit a fixed point. Starting from $x_*\ne\perfect$ with $x_*\preceq\mathsf T(x_*)$, that limit remains above $x_*$ in degradation order and is nonperfect. Descent and the fixed-point gap of \cref{thm:gap} give $U(x_*)\geq\Delta$.

Use a one-sided comparison on $J=L+w-1$ updated check-input pairs: set $x_j=\perfect$ for $j\leq0$, extend constantly by $x_J$ for $j>J$, and use $o_i=\mathsf f(y_i)$ for every $i\geq1$. Start the updated pairs at $(\Delta_0,\Delta_0)$ and reset the right extension after each update. Monotonicity gives a limiting profile, nondecreasing in degradation with spatial position, that dominates the terminated system. If it is nonperfect, its endpoint $x_*=x_J$ satisfies $x_*\preceq\mathsf T(x_*)$. Hold this exterior fixed when varying the profile and put

\begin{equation}\label{eq:profilepotential}
 \mathcal U(z)=\sum_{j=1}^J\mathcal Q(z_j)
 -\sum_{i=1}^J\mathcal P\!\left(\frac1w\sum_{k=0}^{w-1}\mathsf g(z_{i+k})\right).
\end{equation}

Shift the updated coordinates by $(\mathsf Sx)_j=x_{j-1}$, leaving the exterior fixed. The check and variable contributions in the profile potential \cref{eq:profilepotential} telescope. The endpoint potential is at least $\Delta$ by the preceding gap argument, and $\mathcal P(y_0)\geq0$ because it is a sum of conditional entropies. Thus the shift forces a potential decrease of at least $\Delta$:

\begin{equation}\label{eq:upperdrop}
 \mathcal U(\mathsf Sx)-\mathcal U(x)=-U(x_*)-\mathcal P(y_0)\leq-\Delta.
\end{equation}

To bound the same decrease from above, let $v=\mathsf Sx-x$ be the shift direction, $z_s=x+sv$ the interpolation for $0\leq s\leq1$, and $G_{j,\alpha}(s)=D\mathsf g_\alpha(z_{s,j})[v_j]$ the resulting check-output variation. Neighboring fixed-point averages give

\begin{equation}\label{eq:increments}
 v_j=\frac{o_{j-w}-o_j}{w},\qquad
 \epsilon_{j,\alpha}:=H(x_{j,\alpha})-H(x_{j-1,\alpha})\geq0,
 \qquad\sum_j\epsilon_{j,\alpha}\leq1.
\end{equation}

Linearity of signed entropy also gives $\epsilon_{j,\alpha}=[H(o_{j,\alpha})-H(o_{j-w,\alpha})]/w$. Thus a leaf $(o_{j-w,\alpha}-o_{j,\alpha})/w$ has entropy-change magnitude exactly $\epsilon_{j,\alpha}$ before further convolutions, which can only reduce its entropy change. Each component of $v_j$ has TV at most $1/w$. Use the finite-block entropy $h_n$ in the profile potential \cref{eq:profilepotential}, and denote the corresponding fixed-exterior update by $\mathsf T_{n,*}$. For a varied coordinate $j$, the contributing variable indices are $i\in[j-w+1,j]\cap[1,J]$. Omitted left-boundary terms have law $\perfect$ and zero signed entropy; the fixed right exterior has zero variation. The exact residual identity is

\begin{align}\label{eq:profileresidual}
 D\mathcal U_n(z_s)[v]
 &=\sum_{j,\alpha}\omega_\alpha H\bigl((z_{s,j,\alpha}-\mathsf T_{n,*,j,\alpha}(z_s))\star G_{j,\alpha}(s)\bigr),\nonumber\\
 z_s-\mathsf T_{n,*}(z_s)
 &=x-\mathsf T_{n,*}(x)+sv-\int_0^sD\mathsf T_{n,*}(z_\sigma)[v]\,\dd\sigma.
\end{align}

The contribution of $x-\mathsf T_{n,*}(x)$ is bounded in absolute value by $\varepsilon_n\to0$, uniformly in $s$, for this fixed finite profile. Indeed, the block-detector approximation in \cref{eq:blocklimit} gives weak convergence of the residual laws, while convolutions and entropy are jointly continuous. The probability laws in the finitely many signed leaf terms range over a compact set as $s\in[0,1]$, so convergence is uniform in $s$. No rate is inferred: the entropy kernel is continuous but not Lipschitz at perfect information. The $sv$ contribution is nonnegative by the same-order signed-entropy inequality. It remains to bound the integral.

A convolution tree changes entropy by at most $H(b)-H(a)$ when one leaf changes from $a$ to $b\succeq a$. With a second arbitrary signed probability-law difference, the absolute mixed entropy is bounded by the same quantity: fix the second leaf to either law and subtract two numbers in that interval. This follows by iterating the entropy duality in \cref{eq:duality} and degradation monotonicity. Convolutions are also nonexpansive in canonical TV.

Put $d_0=r+g-1$. Each check derivative has $d_0$ leaf terms. Holding the detector law fixed, the variable derivatives have $\ell-1$ and $g-1$ factors. Use $1/w$ at a derivative leaf and sum the ordered entropy changes at the paired leaf in the check-output variation $G_{j,\alpha}$ using the entropy-increment sum in \cref{eq:increments}. The direct contribution is at most $d_0^2[r(\ell-1)+g(g-1)]/w$, as in the BMS proof.

The additional detector term, with $y=y_{i,2}$ and $d=y^{\star g}$, is paired with the check-output variation $G_{j,2}(s)$. Here $y$ and its variation $\dot y_{i,2}$ are evaluated at $z_\sigma$, while $G_{j,2}(s)$ is evaluated at $z_s$. Differentiating the entropy variation formula \cref{eq:variation} once more gives

\begin{equation}\label{eq:detectorpair}
 \begin{aligned}
 &H\bigl(De_n(d)[g y^{\star(g-1)}\star\dot y_{i,2}]
       \star y^{\star(g-1)}\star G_{j,2}(s)\bigr)\\
 &\qquad=D^2h_n(d)[g y^{\star(g-1)}\star\dot y_{i,2},\ y^{\star(g-1)}\star G_{j,2}(s)].
 \end{aligned}
\end{equation}

Each first-direction leaf is a convolution of probability laws with $(o_{k-w,\alpha}-o_{k,\alpha})/w$; it has zero mass and is $1/w$ times a difference of probability laws. Their total weighted count after averaging is $g d_0$, giving total TV at most $g d_0/w$. Each second-direction leaf is a probability convolution of an informative-minus-degraded difference divided by $w$. Degradation preservation and the preceding entropy identity bound its entropy change by $\epsilon_{j,\alpha}$. Hence the mixed-derivative bound in \cref{eq:hessian} applies term by term by bilinearity, with the first difference arbitrary and the second ordered. Summing these entropy changes over $j$ and the $d_0$ leaves contributes a factor of at most $d_0$. The detector bound \eqref{eq:hessian}, followed by the weight $\omega_2=g$, bounds this term by $J_pg^2d_0^2/w$. All averaging weights are nonnegative and sum to at most one; the bound allows different path parameters $s,\sigma$. Combining these estimates, integrating in $s$, and then taking $n\to\infty$ gives

\begin{equation}\label{eq:lowerdrop}
 \frac{\dd}{\dd s}\mathcal U_n(z_s)\geq-\varepsilon_n-\frac{K_p s}{w},
 \qquad \mathcal U(\mathsf Sx)-\mathcal U(x)\geq-\frac{K_p}{2w}.
\end{equation}

The gap-based decrease in \cref{eq:upperdrop} is at least $\Delta$, whereas the derivative-based estimate in \cref{eq:lowerdrop} allows a decrease of at most $K_p/(2w)$. The sufficient width condition in \cref{eq:width} makes the latter strictly smaller than the former, a contradiction. The comparison limit, and therefore the terminated limit, is perfect.

\end{proof}

\section{Finite codes and actual rate}\label{sec:codes}

The threshold-saturation result of \cref{thm:saturation} concerns density evolution. To obtain the finite codes in \cref{thm:main}, we must approximate that recursion by finite decoding neighborhoods and control the dimension after puncturing. Use the balanced ensemble~\cite[Sec.~3.1]{Kasai2026}. An active section contains $M$ transmitted variables of degree $g$ and $rM/\ell$ punctured variables of degree $\ell$. Each section has $M$ checks, each with $r$ punctured and $g$ transmitted sockets. Independently and uniformly partition the sockets of each type in each variable section and each check section into $w$ equal offset groups, then uniformly match corresponding groups. Offset $k$ connects variable section $i$ to check section $i+k$, $0\leq k<w$. Choose $M$ so $rM/\ell,rM/w,gM/w$ are integral. Thus finitely many exposed offsets become independent and uniform as $M\to\infty$; offsets are not prescribed separately at each node. Shorten both variable types outside sections $1,\ldots,L$; at most $(L+w-1)M$ nontrivial checks remain.

Transmit consecutive sections, with independent uniform permutations within each section, through one continuous noise trajectory. The permutation is part of the code realization. The decoder first computes check-to-variable SPA messages, uses all $g$ incoming check messages as each transmitted bit's detector helper, and omits the entire root helper in window BCJR. Outgoing transmitted messages combine the extrinsic result with the other $g-1$ check messages; outgoing punctured messages combine the other $\ell-1$ messages. Decisions use all local evidence, with independent fair tie-breaking. Each window uses only its own physical outputs and the helpers at its nonroot positions; it ignores all outputs outside the window. At physical ends truncate the window, using a stationary left boundary and an all-one right boundary. Windows crossing an internal section boundary use the actual helpers on both sides. Away from section boundaries, the limiting density evolution (DE) at fixed window radius $W$ is the coupled recursion in \cref{eq:coupled}, with the limiting extrinsic map $e$ replaced by the window detector $e_W^{\mathrm{win}}$ of \cref{def:helpers}; letting $W\to\infty$ gives the displayed recursion and bit-posterior laws $y_{i,1}^{\star\ell}$ and $e(y_{i,2}^{\star g})\star y_{i,2}^{\star g}$.

For any fixed extended codeword $(u,v)$, where $u$ is the punctured part and $v$ the transmitted part, replace the physical output by $Y\oplus v$ and reverse the sign of each log-likelihood ratio (LLR) according to its variable bit. Linear checks and variable sums respect this transformation. In BCJR, the transformed prior satisfies $\nu'_k(z)=\nu_k(z\oplus v_k)$, so $\nu'_k(Y_k\oplus v_k\oplus e)=\nu_k(Y_k\oplus e)$: the noise-state likelihood and boundary laws are unchanged. Omitting the root helper preserves this identity. Coupling each fair tie decision with its bit-flipped decision proves equality of bit-error laws with the all-zero experiment.

\begin{lemma}[Finite-code density-evolution limit]\label{lem:finite}

Fix the chain length $L$, coupling width $w$, window radius $W$, and finite iteration count $t$ of the preceding construction. Let $P^{(t,W)}(L,w,M)$ denote its average bit-error probability for either variable type, and let $P^{(t)}_{\mathrm{DE}}(L,w)$ be the corresponding prediction of the infinite-window recursion \cref{eq:coupled}. Then ensemble-average errors converge to finite-window density evolution as $M\to\infty$. For either variable type,

\begin{equation}\label{eq:delimits}
 \lim_{W\to\infty}\lim_{M\to\infty}\E P^{(t,W)}(L,w,M)
 =P^{(t)}_{\mathrm{DE}}(L,w).
\end{equation}

\end{lemma}\begin{proof}

Unfold the root decision through the $t$ parallel iterations. A check message exposes its other $r+g-1$ sockets. A punctured-variable message exposes its other $\ell-1$ checks. A transmitted-variable message exposes its detector window and its other $g-1$ checks; the window exposes at most $2W$ nonroot helper products, each with $g$ check messages from the preceding stages. Continue recursively to the uninformative initialization. Count physical output positions as leaves and retain the dependence inside each window. A dependency path crosses at most four stages per iteration: a variable combination, a detector window, a helper product, and a check update leading to the previous iteration. The final decision and initialization need at most four additional stages. A crude bound is $B=\sum_{k=0}^{4t+4}\kappa^k$, where $\kappa=1+\ell+r+2g+(2W+1)(g+1)$ bounds the fan-in of every stage. Thus $B$ is independent of $M$.

Expose this geometry without using observations. By sampling sockets and physical positions without replacement~\cite{RichardsonUrbanke2001}, the probability of an unintended repeated variable or socket, two distinct windows with centers within $q_M=\lfloor\sqrt M\rfloor$, or a window meeting a section boundary is $O(B^2(q_M+W)/M)$. Uses of a variable within its own detector factor are intentional, not collisions. Section-boundary roots contribute $O(BW/M)$ to the average; their actual decoder may cross the boundary, but no homogeneous-section approximation is asserted there.

For every remaining geometry, order the disjoint windows in physical time. For the noise transition matrix $Q$ of \cref{eq:channel}, stationary distribution $\pi=(1/2,1/2)$, and decay factor $\rho$ of \cref{eq:rho}, the Markov bound $\TV(Q^k(s,\cdot),\pi)\leq\rho^k$ couples their joint noise law to independent stationary window laws with total TV cost $O(B\rho^{q_M-2W})$, uniformly in previously exposed endpoint states. Both approximation errors vanish as $M\to\infty$ with $W,t$ fixed. The decoder symmetry proved above reduces any extended codeword to zero; the disjoint helper branches then yield conditionally independent BMS experiments with exactly the finite-window DE laws. The excluded root helper and destination check prevent immediate feedback. This establishes convergence of ensemble-average errors, without a concentration claim. Then use the window approximation in \cref{eq:uniform} and continuity for the fixed number of updates. For the canonical reliability $V$ defined in \cref{sec:prelim}, bit error is the continuous functional $\E[(1-V)/2]$. This gives the successive limits in \cref{eq:delimits}; the section-size limit precedes the window-radius limit.

\end{proof}

The message-neighborhood method for joint detection and decoding goes back to the ISI analysis~\cite[Sec.~III]{KavcicMaMitzenmacher2003}. Here the physical noise in separated windows remains correlated, so the additional Markov coupling estimate is needed. The present existence argument uses average-error convergence rather than an exponential concentration bound.

It remains to verify actual rate rather than just counting variables and checks. Write the extended parity-check matrix as $[A\ B]$, with $A$ acting on the punctured bits, and let $k_{\mathrm{ext}},k_{\mathrm{tx}}$ be the extended and transmitted-projection dimensions. For $n=LM$ and $n_u=Rn$, linear algebra gives

\begin{equation}\label{eq:rank}
 k_{\mathrm{tx}}=k_{\mathrm{ext}}-\dim\ker A,\qquad
 k_{\mathrm{ext}}\geq n\left(R-\frac{w-1}{L}\right).
\end{equation}

The kernel term in the projection-dimension formula \cref{eq:rank} measures punctured degrees of freedom invisible in transmission. To bound it without assuming independent parity checks, test the same graph on an auxiliary BEC of capacity $\alpha>R$. Let $(U,V)$ be uniform on the extended code, with punctured vector $U$ and transmitted vector $V$, and let $Y_{\mathrm{BEC}}$ be the BEC observation of $V$. These codeword vectors are local to the rate argument; they are distinct from the potential $U$ and scalar reliability $V$ used earlier. If its average punctured- and transmitted-bit errors under the uniform extended code are $e_u,e_v$, binary Fano and the BEC capacity bound give

\begin{equation}\label{eq:rate}
 \begin{aligned}
 \dim\ker A&=H(U\mid V)\leq H(U\mid Y_{\mathrm{BEC}})\leq n_u h_2(e_u),\\
 R-\frac{w-1}{L}-R h_2(e_u)&\leq\frac{k_{\mathrm{tx}}}{n}\leq\alpha+h_2(e_v).
 \end{aligned}
\end{equation}

Here $V$ is uniform on the transmitted projection, and the upper rate bound follows from $H(V)\leq n\alpha+H(V\mid Y_{\mathrm{BEC}})$. No rank-independence claim is used.

\begin{proof}[Proof of the finite-code capacity criterion, \cref{thm:main}]

Take auxiliary BEC capacities $\alpha_j\downarrow R$, approaching the design rate. Choose coupling widths $w_j$ large enough for both the Markov-channel saturation theorem, \cref{thm:saturation}, and BMS saturation on that BEC~\cite[Corollary~7.4]{Kasai2026}; next choose chain lengths $L_j$ with $w_j/L_j\to0$ to remove the termination-rate loss. Choose finite iteration counts $t_j$ making both variable-type DE errors small on both channels. Choose window radii $W_j$ using the detector approximation in \cref{eq:uniform}, and section sizes $M_j$ using the finite-code DE limit of \cref{lem:finite} and the memoryless computation-tree limit. The expected sum of the four nonnegative errors tends to zero, so some graph and physical permutation have sum no larger than that expectation. Select such a realization at each $j$. The lower and upper rate bounds in \cref{eq:rate} then both tend to $R$.

Choose an information set in the transmitted projection and a linear section into the extended code to obtain a deterministic encoder. Additive symmetry and symmetric tie-breaking make decoder bit errors identical for every fixed extended codeword, so restricting to this section preserves them. The information-bit error is at most $n/k_{\mathrm{tx}}$ times the transmitted-bit error and therefore vanishes. 

\end{proof}

A two-state BCJR pass for one radius-$W$ root window costs $O(W)$; a direct implementation of $t$ iterations costs $O(tn(W+\ell))$. The theorem chooses finite parameters for each target accuracy and then increases section size. It does not give a useful finite-length or iteration bound.

\section{Conclusion and extensions}\label{sec:extensions}

The potential decomposition in \cref{eq:reduction} and the positive-gap theorem, \cref{thm:gap}, transfer BMS positivity to Markov noise. The posterior contraction in \cref{lem:contraction} supplies the detector bound needed for the saturation theorem, \cref{thm:saturation}. For any finite binary additive-noise block, independent BMS helpers still give BMS extrinsic experiments and $\corr_n\geq\max\{0,C_n-C(e_n(d))\}$ by the entropy-correction argument in \cref{lem:correction}. Extending the theorem requires entropy and detector limits, a uniform summable interaction bound, and a computation-tree approximation.

For binary Markov noise with a strictly positive transition matrix $Q$, the posterior-transition row calculation in \cref{eq:twostate}, maximized over likelihood ratios, gives $\eta(Q)=|\sqrt{Q_{00}Q_{11}}-\sqrt{Q_{01}Q_{10}}|/(\sqrt{Q_{00}Q_{11}}+\sqrt{Q_{01}Q_{10}})<1$. A stationary two-state chain is reversible, so the same bound applies backward. The second eigenvalue alone does not give this posterior bound. Unequal switching probabilities also change the capacity and trivial point, generally to $(\Delta_0,e(\Delta_0))$. These replacements and the limit and realization arguments remain to be checked; the stated theorem concerns the symmetric case.

For hidden Markov additive noise, including Gilbert--Elliott models, the error bit need not be Markov. A strictly positive hidden-state transition matrix suggests using projective contraction~\cite{Birkhoff1957}. One must obtain a bound uniform over all helper laws, including perfect helpers, and derive the mixed-entropy estimate and finite-window approximation. Ordinary state mixing alone does not provide the required posterior control.

General finite-state channels may lose additive symmetry or have input-dependent transitions. They require a suitable symmetry construction and a dependence bound compatible with the decoder; symmetric information rate must also be distinguished from capacity. Insertions and deletions further require alignment and rate-normalization arguments. Quantitative convergence at moderate widths and finite iteration counts remains open here, and other degree families need corresponding BMS positivity results.

\begingroup\small\setlength{\baselineskip}{13pt}
\bibliographystyle{IEEEtran}
\bibliography{references}

\begin{thebibliography}{10}
\providecommand{\url}[1]{#1}
\csname url@samestyle\endcsname
\providecommand{\newblock}{\relax}
\providecommand{\bibinfo}[2]{#2}
\providecommand{\BIBentrySTDinterwordspacing}{\spaceskip=0pt\relax}
\providecommand{\BIBentryALTinterwordstretchfactor}{4}
\providecommand{\BIBentryALTinterwordspacing}{\spaceskip=\fontdimen2\font plus
\BIBentryALTinterwordstretchfactor\fontdimen3\font minus
  \fontdimen4\font\relax}
\providecommand{\BIBforeignlanguage}[2]{{%
\expandafter\ifx\csname l@#1\endcsname\relax
\typeout{** WARNING: IEEEtran.bst: No hyphenation pattern has been}%
\typeout{** loaded for the language `#1'. Using the pattern for}%
\typeout{** the default language instead.}%
\else
\language=\csname l@#1\endcsname
\fi
#2}}
\providecommand{\BIBdecl}{\relax}
\BIBdecl

\bibitem{OkazakiKasai2014}
\BIBentryALTinterwordspacing
T.~Okazaki and K.~Kasai, ``Spatially-coupled {MacKay--Neal} codes with no bit
  nodes of degree two achieve the capacity of {BEC},'' in \emph{2014 IEEE
  International Symposium on Information Theory}, 2014, pp. 506--510. [Online].
  Available: \url{https://doi.org/10.1109/ISIT.2014.6874884}
\BIBentrySTDinterwordspacing

\bibitem{FukushimaOkazakiKasai2015}
\BIBentryALTinterwordspacing
M.~Fukushima, T.~Okazaki, and K.~Kasai, ``Spatially-coupled {MacKay--Neal}
  codes universally achieve the symmetric information rate of arbitrary
  generalized erasure channels with memory,'' in \emph{2015 IEEE International
  Symposium on Information Theory}, 2015, pp. 899--903, preprint:
  arXiv:1501.06736v1. [Online]. Available:
  \url{https://arxiv.org/abs/1501.06736}
\BIBentrySTDinterwordspacing

\bibitem{Gilbert1960}
\BIBentryALTinterwordspacing
E.~N. Gilbert, ``Capacity of a burst-noise channel,'' \emph{Bell System
  Technical Journal}, vol.~39, no.~5, pp. 1253--1265, Sep. 1960. [Online].
  Available:
  \url{https://vtda.org/pubs/BSTJ/vol39-1960/articles/bstj39-5-1253.pdf}
\BIBentrySTDinterwordspacing

\bibitem{Elliott1963}
\BIBentryALTinterwordspacing
E.~O. Elliott, ``Estimates of error rates for codes on burst-noise channels,''
  \emph{Bell System Technical Journal}, vol.~42, no.~5, pp. 1977--1997, Sep.
  1963. [Online]. Available:
  \url{https://doi.org/10.1002/j.1538-7305.1963.tb00955.x}
\BIBentrySTDinterwordspacing

\bibitem{MushkinBarDavid1989}
\BIBentryALTinterwordspacing
M.~Mushkin and I.~Bar-David, ``Capacity and coding for the {Gilbert--Elliott}
  channels,'' \emph{IEEE Transactions on Information Theory}, vol.~35, no.~6,
  pp. 1277--1290, Nov. 1989. [Online]. Available:
  \url{https://web.mit.edu/6.441/www/reading/IT-V35-N6.pdf}
\BIBentrySTDinterwordspacing

\bibitem{BCJR1974}
\BIBentryALTinterwordspacing
L.~R. Bahl, J.~Cocke, F.~Jelinek, and J.~Raviv, ``Optimal decoding of linear
  codes for minimizing symbol error rate,'' \emph{IEEE Transactions on
  Information Theory}, vol.~20, no.~2, pp. 284--287, 1974. [Online]. Available:
  \url{https://doi.org/10.1109/TIT.1974.1055186}
\BIBentrySTDinterwordspacing

\bibitem{PfisterSiegel2008}
\BIBentryALTinterwordspacing
H.~D. Pfister and P.~H. Siegel, ``Joint iterative decoding of {LDPC} codes for
  channels with memory and erasure noise,'' \emph{IEEE Journal on Selected
  Areas in Communications}, vol.~26, no.~2, pp. 320--337, 2008. [Online].
  Available: \url{https://doi.org/10.1109/JSAC.2008.080209}
\BIBentrySTDinterwordspacing

\bibitem{ArnoldLoeliger2001}
\BIBentryALTinterwordspacing
D.~Arnold and H.-A. Loeliger, ``On the information rate of binary-input
  channels with memory,'' in \emph{2001 IEEE International Conference on
  Communications}, vol.~9, Jun. 2001, pp. 2692--2695. [Online]. Available:
  \url{https://doi.org/10.1109/ICC.2001.936639}
\BIBentrySTDinterwordspacing

\bibitem{PfisterSoriagaSiegel2001}
\BIBentryALTinterwordspacing
H.~D. Pfister, J.~B. Soriaga, and P.~H. Siegel, ``On the achievable information
  rates of finite state {ISI} channels,'' in \emph{2001 IEEE Global
  Telecommunications Conference}, vol.~5, 2001, pp. 2992--2996. [Online].
  Available: \url{https://pfister.ee.duke.edu/papers/globecom2001.pdf}
\BIBentrySTDinterwordspacing

\bibitem{KavcicMaMitzenmacher2003}
\BIBentryALTinterwordspacing
A.~Kav{\v c}i{\'c}, X.~Ma, and M.~Mitzenmacher, ``Binary intersymbol
  interference channels: {Gallager} codes, density evolution, and code
  performance bounds,'' \emph{IEEE Transactions on Information Theory},
  vol.~49, no.~7, pp. 1636--1652, 2003. [Online]. Available:
  \url{https://www.eecs.harvard.edu/~michaelm/postscripts/toit2003a.pdf}
\BIBentrySTDinterwordspacing

\bibitem{EckfordKschischangPasupathy2005}
\BIBentryALTinterwordspacing
A.~W. Eckford, F.~R. Kschischang, and S.~Pasupathy, ``Analysis of low-density
  parity-check codes for the {Gilbert--Elliott} channel,'' \emph{IEEE
  Transactions on Information Theory}, vol.~51, no.~11, pp. 3872--3889, 2005.
  [Online]. Available: \url{https://doi.org/10.1109/TIT.2005.856934}
\BIBentrySTDinterwordspacing

\bibitem{FelstromZigangirov1999}
\BIBentryALTinterwordspacing
A.~Jim{\'e}nez~Felstr{\"o}m and K.~S. Zigangirov, ``Time-varying periodic
  convolutional codes with low-density parity-check matrix,'' \emph{IEEE
  Transactions on Information Theory}, vol.~45, no.~6, pp. 2181--2191, 1999.
  [Online]. Available: \url{https://doi.org/10.1109/18.782171}
\BIBentrySTDinterwordspacing

\bibitem{KudekarRichardsonUrbanke2011}
\BIBentryALTinterwordspacing
S.~Kudekar, T.~J. Richardson, and R.~L. Urbanke, ``Threshold saturation via
  spatial coupling: Why convolutional {LDPC} ensembles perform so well over the
  {BEC},'' \emph{IEEE Transactions on Information Theory}, vol.~57, no.~2, pp.
  803--834, 2011. [Online]. Available:
  \url{https://doi.org/10.1109/TIT.2010.2095072}
\BIBentrySTDinterwordspacing

\bibitem{KudekarRichardsonUrbanke2013}
\BIBentryALTinterwordspacing
S.~Kudekar, T.~Richardson, and R.~L. Urbanke, ``Spatially coupled ensembles
  universally achieve capacity under belief propagation,'' \emph{IEEE
  Transactions on Information Theory}, vol.~59, no.~12, pp. 7761--7813, 2013.
  [Online]. Available: \url{https://doi.org/10.1109/TIT.2013.2280915}
\BIBentrySTDinterwordspacing

\bibitem{Kumar2014}
\BIBentryALTinterwordspacing
S.~Kumar, A.~J. Young, N.~Macris, and H.~D. Pfister, ``Threshold saturation for
  spatially coupled {LDPC} and {LDGM} codes on {BMS} channels,'' \emph{IEEE
  Transactions on Information Theory}, vol.~60, no.~12, pp. 7389--7415, 2014.
  [Online]. Available: \url{https://doi.org/10.1109/TIT.2014.2360692}
\BIBentrySTDinterwordspacing

\bibitem{KudekarKasai2011}
\BIBentryALTinterwordspacing
S.~Kudekar and K.~Kasai, ``Threshold saturation on channels with memory via
  spatial coupling,'' in \emph{2011 IEEE International Symposium on Information
  Theory}, 2011, preprint: arXiv:1102.0406v1. [Online]. Available:
  \url{https://arxiv.org/abs/1102.0406}
\BIBentrySTDinterwordspacing

\bibitem{NguyenEtAl2012}
\BIBentryALTinterwordspacing
P.~S. Nguyen, A.~Yedla, H.~D. Pfister, and K.~R. Narayanan, ``Threshold
  saturation of spatially-coupled codes on intersymbol-interference channels,''
  in \emph{2012 IEEE International Conference on Communications}, 2012, pp.
  2181--2186. [Online]. Available:
  \url{https://doi.org/10.1109/ICC.2012.6364433}
\BIBentrySTDinterwordspacing

\bibitem{MacKayNeal1995}
\BIBentryALTinterwordspacing
D.~J.~C. MacKay and R.~M. Neal, ``Good codes based on very sparse matrices,''
  in \emph{Cryptography and Coding}, ser. Lecture Notes in Computer Science,
  vol. 1025.\hskip 1em plus 0.5em minus 0.4em\relax Springer, 1995, pp.
  100--111. [Online]. Available: \url{https://doi.org/10.1007/3-540-60693-9_13}
\BIBentrySTDinterwordspacing

\bibitem{MacKay1999}
\BIBentryALTinterwordspacing
D.~J.~C. MacKay, ``Good error-correcting codes based on very sparse matrices,''
  \emph{IEEE Transactions on Information Theory}, vol.~45, no.~2, pp. 399--431,
  1999. [Online]. Available: \url{https://doi.org/10.1109/18.748992}
\BIBentrySTDinterwordspacing

\bibitem{KasaiSakaniwa2011}
\BIBentryALTinterwordspacing
K.~Kasai and K.~Sakaniwa, ``Spatially-coupled {MacKay--Neal} codes and
  {Hsu--Anastasopoulos} codes,'' \emph{IEICE Transactions on Fundamentals of
  Electronics, Communications and Computer Sciences}, vol. E94-A, no.~11, pp.
  2161--2168, 2011. [Online]. Available:
  \url{https://doi.org/10.1587/transfun.E94.A.2161}
\BIBentrySTDinterwordspacing

\bibitem{Kasai2026}
\BIBentryALTinterwordspacing
K.~Kasai, ``Spatially coupled {MacKay--Neal} codes achieve capacity on {BMS}
  channels at fixed degrees,'' 2026, arXiv preprint, arXiv:2609.15526v1.
  [Online]. Available: \url{https://arxiv.org/abs/2609.15526v1}
\BIBentrySTDinterwordspacing

\bibitem{VerduHan1994}
\BIBentryALTinterwordspacing
S.~Verd{\'u} and T.~S. Han, ``A general formula for channel capacity,''
  \emph{IEEE Transactions on Information Theory}, vol.~40, no.~4, pp.
  1147--1157, Jul. 1994. [Online]. Available:
  \url{https://doi.org/10.1109/18.335960}
\BIBentrySTDinterwordspacing

\bibitem{Alajaji1995}
\BIBentryALTinterwordspacing
F.~Alajaji, ``Feedback does not increase the capacity of discrete channels with
  additive noise,'' \emph{IEEE Transactions on Information Theory}, vol.~41,
  no.~2, pp. 546--549, Mar. 1995. [Online]. Available:
  \url{https://doi.org/10.1109/18.370168}
\BIBentrySTDinterwordspacing

\bibitem{Watanabe1960}
\BIBentryALTinterwordspacing
S.~Watanabe, ``Information theoretical analysis of multivariate correlation,''
  \emph{IBM Journal of Research and Development}, vol.~4, no.~1, pp. 66--82,
  1960. [Online]. Available: \url{https://doi.org/10.1147/rd.41.0066}
\BIBentrySTDinterwordspacing

\bibitem{RichardsonUrbanke2001}
\BIBentryALTinterwordspacing
T.~J. Richardson and R.~L. Urbanke, ``The capacity of low-density parity-check
  codes under message-passing decoding,'' \emph{IEEE Transactions on
  Information Theory}, vol.~47, no.~2, pp. 599--618, 2001. [Online]. Available:
  \url{https://doi.org/10.1109/18.910577}
\BIBentrySTDinterwordspacing

\bibitem{Birkhoff1957}
\BIBentryALTinterwordspacing
G.~Birkhoff, ``Extensions of {Jentzsch}'s theorem,'' \emph{Transactions of the
  American Mathematical Society}, vol.~85, no.~1, pp. 219--227, 1957. [Online].
  Available: \url{https://doi.org/10.1090/S0002-9947-1957-0087058-6}
\BIBentrySTDinterwordspacing

\end{thebibliography}
\endgroup
\end{document}